\documentclass[11pt,table]{article}
\pdfoutput=1

\usepackage{authblk}
\usepackage{microtype}
\usepackage{graphicx}
\usepackage{booktabs} 
\usepackage{tikz}
\usetikzlibrary {positioning}
\usepackage{hyperref}
\usepackage{mathtools}

\title{Finding Equilibrium in Multi-Agent Games \\with Payoff Uncertainty}
\author[1]{Wenshuo Guo\thanks{wsguo@berkeley.edu}}
\author[1]{Mihaela Curmei}
\author[1]{Serena Wang}
\author[1]{Benjamin Recht}
\author[1,2]{Michael I. Jordan}
\affil[1]{EECS, University of California, Berkeley}
\affil[2]{Department of Statistics, University of California, Berkeley}

\usepackage[utf8]{inputenc} 
\usepackage[T1]{fontenc}    
\usepackage[in]{fullpage}
\usepackage{booktabs}       
\usepackage{amsfonts}       
\usepackage{nicefrac}       
\usepackage{microtype}      
\usepackage{amsmath,amssymb}
\usepackage{amsthm}
\usepackage{mathtools}
\usepackage{mathrsfs}
\usepackage{thmtools}
\usepackage{thm-restate}
\usepackage{comment}
\usepackage{parskip}
\usepackage{siunitx}
\usepackage{etoolbox}
\usepackage{enumitem}
\usepackage{graphicx}
\usepackage{subcaption}
\usepackage[textsize=scriptsize,textwidth=2.2cm]{todonotes}
\usepackage{tabularx}
\usepackage{xcolor} 
\usepackage[numbers,sort]{natbib}
\usepackage{hyperref}
\usepackage{longtable}
\usepackage{afterpage}
\usepackage{algorithm}
\usepackage{algpseudocode}
\usepackage{bm}
\usepackage{etoc}

\extrafloats{100}

\newtoggle{showtodos}
\togglefalse{showtodos}

\newtoggle{showappendix}
\toggletrue{showappendix}

\newtoggle{isicml}
\togglefalse{isicml}

\newcommand{\R}{\mathbb{R}}
\newcommand{\E}{\mathbb{E}}

\hypersetup{draft}
\algblock{Input}{EndInput}
\algnotext{EndInput}
\algblock{Output}{EndOutput}
\algnotext{EndOutput}

\usepackage[title]{appendix}

\usepackage[utf8]{inputenc} 
\usepackage[T1]{fontenc}    
\usepackage{hyperref}       
\usepackage{url}            
\usepackage{booktabs}       
\usepackage{amsfonts}       
\usepackage{nicefrac}       
\usepackage{microtype}      
\usepackage{graphicx}
\usepackage{amsthm}
\usepackage{amsmath,amssymb, graphicx,multicol,multirow,array, bbm,varwidth,semantic}
\usepackage[margin=1in]{geometry}
\hypersetup{
    colorlinks=true,
    linkcolor=blue,
    filecolor=magenta,      
    urlcolor=cyan,
}
\usepackage{xcolor}
\usepackage{thmtools, thm-restate}

\makeatletter

\theoremstyle{definition}
\newtheorem{definition}{Definition}
\newtheorem{lemma}{Lemma}
\declaretheorem{theorem}

\usepackage[section]{placeins}
\usepackage{booktabs} 
\usepackage{multirow}
\usepackage{siunitx}

\theoremstyle{definition}
\theoremstyle{remark}

\usepackage{booktabs}

\usepackage{tikz}
\usetikzlibrary {positioning}

\DeclareMathOperator*{\argmax}{argmax}
\newcommand{\x}{\times}

\newcommand{\Prob}{\mathbb{P}}
\renewcommand{\bar}{\overline}
\renewcommand{\epsilon}{\varepsilon}

\begin{document}

\date{}

\maketitle

\begin{abstract}
We study the problem of finding equilibrium strategies in multi-agent games with incomplete payoff information, where the payoff matrices are only known to the players up to some bounded uncertainty sets. In such games, an \textit{ex-post equilibrium} characterizes equilibrium strategies that are robust to the payoff uncertainty. When the game is one-shot, we show that in zero-sum polymatrix games, an ex-post equilibrium can be computed efficiently using linear programming. We further extend the notion of ex-post equilibrium to stochastic games, where the game is played repeatedly in a sequence of stages and the transition dynamics are governed by an Markov decision process (MDP). We provide sufficient condition for the existence of an ex-post Markov perfect equilibrium (MPE). We show that under bounded payoff uncertainty, the value of any two-player zero-sum stochastic game can be computed up to a tight value interval using dynamic programming.
\end{abstract}

\section{Introduction}

The problem of certifying the existence of equilibria of multi-agent games has been a focal point in the game theory literature. In particular, seminal work has resulted in algorithms that find equilibrium strategies efficiently when all parameters of the game are known to the players (a.k.a., games of \textit{complete information}) \cite{nash:1950, shapley:1953}.

In practice, the parameters of the game are often only known partially by the players. This presents a core problem: how can we find equilibria in games in a way that is robust to uncertainty in the game parameters? The problem of finding robust solutions to multi-player games is a difficult area with many open questions, starting with even the most fundamental problem of defining robustness in games.

In this work, we focus on multi-agent games with payoff uncertainty. We consider both one-shot zero-sum games and stochastic games \cite{shapley:1953}, where the game dynamics are driven by a Markov decision process (MDP). We assume that the payoffs live in a bounded uncertainty set, with no further assumptions on the distribution of the payoffs.
We focus on the ex-post equilibrium, which is a common distribution-free concept first introduced in auction theory \cite{holmstrom:1983, cremer:1985}. The ex-post equilibrium is a strict notion of equilibrium, where each player's
strategy is a best response to the other players' strategies, under all possible realizations of the uncertain payoffs.

In one-shot games, we consider \textit{finite}, \textit{simultaneous-move} games. In such a game, there are a finite number of players who each may take a finite number of actions. The player chooses their actions simultaneously, in the absence of knowledge of the actions chosen by the other players. In fact, although Nash equilibria always exist in such a game, even with complete information, finding Nash equilibria for general-sum finite games can be intractable \cite{daskalakis:2006}. Therefore, we restrict to a smaller classes of finite games: zero-sum polymatrix games. In such games of complete information, a Nash equilibrium can be found efficiently using linear programming \cite{cai:2016}. 

We extend the idea of ex-post equilibria to stochastic multi-agent games \cite{shapley:1953}. We consider the infinite discounted setting, where the game is played repeatedly for an infinite number of stages, and the transition dynamics among stages are captured by an MDP. The goal of each player is to maximize the discounted expected payoff. A Markov perfect equilibrium (MPE) refers to a perfect equilibrium where the players' strategies depend only on the current stage. Under complete information about the stage payoffs and the MDP, \citet{shapley:1953} showed that an MPE can be found using dynamic programming. We consider payoff uncertainty in stochastic games. Specifically, we assume that at each stage, the payoffs live in some bounded uncertainty set. 

\subsection*{Contributions:} 
\begin{enumerate}[noitemsep, topsep=0em]
    \item We show that for zero-sum polymatrix games with bounded payoff uncertainty sets, an ex-post equilibrium can be found efficiently using linear programming. This is different from the work of \citet{Aghassi:2006}, who work with a weaker notion of robust equilibrium in general-sum finite games, one that does not yield efficiency guarantees.
    \item We  extend the notion of ex-post equilibrium to multi-agent stochastic games with payoff uncertainties; i.e., we define a notion of ex-post MPE. We provide sufficient conditions for the existence of such an equilibrium, and show that the feasible values of any two-player zero-sum stochastic game can be bounded in a tight interval using dynamic programming.
\end{enumerate}



\section{Related Work}

\textbf{Payoff uncertainty in one-shot games:} Multiple techniques have emerged for handling games with incomplete information in one-shot games. \citet{Harsanyi:1967} first extended Nash's result to games with incomplete information by assuming that the payoffs are drawn from a probability distribution known to all players, and introduced the concept of a \textit{Bayesian equilibrium}. 

While \citet{Harsanyi:1967}'s method assumes complete common knowledge only of the distribution of the payoffs, this assumption may still be too strong in practice. \citet{Aghassi:2006} consider the case when the distribution of the payoffs is unknown, but the payoffs are drawn from a bounded uncertainty set. They describe two ``distribution free'' concepts of equilibrium: the ex-post equilibrium and the robust-optimization equilibrium, where the robust-optimization equilibrium is a less restrictive notion that is guaranteed to exist. They give a gradient-based algorithm to approximately find a robust-optimization equilibrium for a general-sum finite game when the payoff comes from a polyhedral uncertainty set. However, their algorithm does not come with complexity or optimality guarantees. Games with payoff uncertainties have also been explored in applications in operational research, multi-agent system and controls \cite{piliouras2016risk, miao2018hybrid, zhang2019playing, garnaev2019power, rowland2019multiagent}.

\textbf{Stochastic games:}
 Many real-world systems can be modeled as multi-agent systems and are dynamical in nature. There has been extensive study on stochastic games in various fields including economics, mathematics and operation research, etc \cite{shapley:1953, mertens:1981, condon1992complexity, neyman2003stochastic}. Stochastic games are also closely related to the literature on reinforcement learning (RL). In a classic RL setting, MDPs are used to model a single agent’s interaction with the environment. In stochastic games, the MDP is generalized to a repeated game setting where multiple agents simultaneously interact with the environment.

Various types of incomplete information have been considered in repeated games. \citet{sastry1994decentralized, kiekintveld2011approximation} consider stochastic payoffs, where the payoffs are drawn from unknown distributions. \citet{kardecs:2011} consider N-player discounted stochastic games where both the transition probabilities and payoffs of the game uncertain and bounded in some uncertainty set. They show the existence of a robust-optimization equilibrium and propose an mathematical programming formulation to compute it. However, the robust-optimization equilibrium is a weaker notion than the ex-post equilibrium, and their algorithm is not guaranteed to find an ex-post MPE. 

\section{Ex-post Equilibrium in Polymatrix Games}
\label{sec:polymatrix}
\subsection{Polymatrix games}
An $N$-player zero-sum polymatrix game \cite{cai:2016} can be informally expressed as a graph, where each node represents a player, and each edge represents a two-player game between two nodes. The payoff for a given player is the sum of the player's payoffs from each game on an edge connected to the player's node. The $N$-player polymatrix game is \textit{zero-sum} if the sum of the total payoffs for all players is zero. 

Formally, an $N$-player polymatrix game \cite{cai:2016} consists of the following:
\begin{itemize}
    \item A finite set $[N] = \{1,...,N\}$ of players (nodes), and a finite set $E$ of edges, which consists of unordered pairs $[i,j]$ of players, $i \neq j, i \in [N], j \in [N]$.
    \item For each player $i \in [N]$, a finite set of actions $[d_i] = \{1,...,d_i\}$.
    
    \item A mixed strategy for player $i$ is defined by a vector in the simplex $x^i \in \Delta_{d_i}$, where $\Delta_{d_i} = \{x \in \mathbb{R}^{d_i}: x^T \textbf{1} = 1$, $x \geq 0\}$. Each entry of $x^i$ is the probability that player 1 plays action $i$. When there exists one entry of $x^i$ that is equal to one, $x^i$ is a pure strategy.
    
    \item For each edge $[i,j] \in E$, a two-player general-sum game with payoff matrices $A^{ij} \in \mathbb{R}^{d_i \times d_j}, A^{ji} \in \mathbb{R}^{d_j \times d_i}$. If player $i$ chooses $k \in [d_i]$ and player $j$ chooses action $l \in [d_j]$, then $A^{ij}_{kl}$ is the payoff for player $i$ and $A^{ji}_{lk}$ is the payoff for player $j$ in the two-player game. Note that the player chooses one strategy and plays it in all pairwise games with other players.

\end{itemize}

For each player $i \in [N]$ and a strategy profile for all players $(x^1,...,x^N)$, the total payoff for player $i$ is the sum of the payoffs over adjacent edges:
$$p_i(x^1,...,x^N) = \sum_{j: j \in [N], [i,j] \in E} (x^i)^T A^{ij} x^j.$$ Note that for a fixed $i$, the summation of the payoffs is taken over all edges in $E$ that include the player $i$ (and not the summation over the full set of edges in $E$). 

Now we are ready to define a zero-sum polymatrix game.
\begin{definition} \label{def: nzerosum} (Zero-sum polymatrix game)
An $N$-player polymatrix game is \textit{zero-sum} if for all pure strategy profiles $(x_1,...,x_N)$, the sum of the total payoffs for all players is zero: $$\sum_{i=1}^N p_i(x^1,...,x^N) = 0.$$
\end{definition}

A zero-sum polymatrix game is a more general game than a two-player zero-sum game. With complete information about the payoff matrices and finite action space, a Nash equilibrium exists for any zero-sum polymatrix game, and can be found efficiently using linear programming \cite{cai:2016}. We provide further definitions of the Nash equilibrium for polymatrix games and two-player zero-sum games in Appendix \ref{app:definitions}.

\subsection{Solving ex-post equilibrium in polymatrix games}

A Nash equilibrium is well-defined in polymatrix games when the payoff matrices are fully known. However, when the payoffs are uncertain, we need a more generalized notion of equilibrium to handle the uncertainty. An ex-post equilibrium describes a ``distribution-free" equilibrium that is robust to the payoff uncertainties \cite{holmstrom:1983, cremer:1985, Aghassi:2006}.  The idea is that for all feasible payoffs, no player has the incentive to change the strategy on their own.

\begin{definition}
(Ex-post equilibrium for $N$-player polymatrix game) A mixed strategy profile $(x^{1*},...,x^{N*})$ is an \emph{ex-post equilibrium} for an $N$-player polymatrix game with uncertainty sets $\mathcal{U}^{ij}$ ($A^{ij} \in \mathcal{U}^{ij}$ for all $i,j \in \{1,...,N\}$) if for each player $i \in \{1,...,N\}$, we have
$$x^{i*} \in \arg\max_{x^i \in \Delta_{d_i}} \sum_{[i,j] \in E} (x^i)^T A^{ij} x^{j*} \quad \forall A^{ij} \in \mathcal{U}^{ij}.$$
\end{definition}

Since ex-post equilibrium is a restrictive notion of equilibrium, it is not guaranteed to always exist \cite{Aghassi:2006}. However, when it exists, the ex-post equilibrium defines a powerful set of equilibrium strategies which are robust to payoff uncertainties. 
In a two-player zero-sum game, an ex-post equilibrium also allows us to conveniently locate the value of the game given the uncertainty set. Given an ex-post equilibrium $(x^\ast,y^\ast)$, for any feasible payoff $A \in \mathcal{U}$, the value of the game is simply contained in the set $\{(x^\ast)^\top A (y^\ast)^\top: A \in \mathcal{U}\}$. 

We show in Theorem \ref{thm:LP-expost-n} that any existing ex-post equilibrium can be found exactly for an $N$-player zero-sum polymatrix game using linear programming. 

To build the linear programming representation for Theorem \ref{thm:LP-expost-n}, we adapt an idea due to \citet{cai:2016} to characterize the payoffs of a zero-sum polymatrix game using a square matrix $R$ as follows: $R \in \mathbb{R}^{\sum_{i=1}^N d_i \times \sum_{i=1}^N d_i}$. The rows of $R$ are indexed by (player: action) pairs $(i:a_i)$ for $i \in [N]$, $a_i\in [d_i]$. The columns of $R$ are indexed in the same way. The entry $R_{(i:a_i),(j:a_j)} = A^{ij}_{a_ia_j}$ if $[i,j] \in E$. Therefore, $R$ characterizes the full payoff information of the game (see Appendix \ref{app:r_matrix} for a visual schematic of $R$). 

We consider the uncertainty set as the convex hull of known payoff matrices: $R \in \textbf{conv} \{R_1, ...R_K\}$. This is equivalent to having each pairwise payoff matrix be in a convex hull of matrices. For each $R_i$ where $i \in [K]$, the corresponding game is zero-sum as defined in Definition \ref{def: nzerosum}.  The robust LP formulation for the multi-player zero-sum polymatrix game is as follows:
\begin{align}\label{multiLP}
    \begin{split}
      \min_{x,w}\;\; &\sum_{l=1}^K\sum_{i=1}^N w_i^l \\
      \text{s.t.}  \;\; & w_i^l \geq e_{a_i}^T R_l x \;\;\;\;\;\; \forall i \in [N], a_i \in [d_i], l \in [K] \\
      &x \in \Delta,
    \end{split}
\end{align}
where $x \in \mathbb{R}^{\sum_{i \in [N]}d_i}$ is the vector which represents the strategy profile for all players, concatenating $x^i, i \in [N]$. The strategy for each player $i$ is in the probability simplex: $x^i \in \Delta_{d_i}$. We use $e_{a_i} \in \mathbb{R}^{\sum_{i \in [N]}d_i}$ to represent a one-hot vector, where player $i$ playing a single action $a_i$ with probability 1, and all other entries of $e_{a_i}$ are zero.


\begin{restatable}[Proof in Appendix \ref{app:proof-polymatrix}]{theorem}{LPexpostn}
\label{thm:LP-expost-n}
 Any ex-post equilibrium of the zero-sum polymatrix game with uncertainty set $\mathcal{U} = \textbf{conv}{\{R_1, R_2,...R_K\}}$ is an optimal solution of the LP problem (\ref{multiLP}). Conversely, any optimal solution for the LP problem (\ref{multiLP}) is an ex-post equilibrium of the game. 
\end{restatable}

We provide further discussion on the existence of ex-post equilibrium and characterize the \textit{maximal uncertainty sets} for an ex-post equilibrium to exist in Appendix \ref{app:maximal-uncertainty}. 

\section{Stochastic Games with Payoff Uncertainty}
\label{sec:stochastic}

We extend the discussion on payoff uncertainties to discounted stochastic games with infinite horizon. We further extend the notion of ex-post equilibrium to ex-post MPE in such games. We provide sufficient conditions for the existence of such ex-post MPE, and show that the feasible value of any two-player zero-sum stochastic games with payoff uncertainties can be located up to a value interval using dynamic programming.

A stochastic game \cite{shapley:1953} with $N$ players is defined as follows. Denote $\mathcal{S}$ as the stage space, and $\mathcal{X}_i(s)$ as the set of actions that are available to player $i$ in stage $s$. We consider finite stage and action spaces. 
For each stage $s$, define the stage action profile to be an element of the product space of action sets available to all players $i$ in stage $s$: $\textbf{x} \in \Pi_i \mathcal{X}_i(s)$. Denote by $M_i(\textbf{x}, s)$ the stage payoff for player $i$, and let $\Prob(s'|s, \textbf{x})$ denote the transition probability which is a distribution on the stage space $\mathcal{S}$. 


The game starts in an initial stage $s^0$. At each stage $t$, all players simultaneously choose their actions according to their strategies, the stage payoffs are realized, and the game transitions to the next stage according to the transition probability. Let $\mathcal{H}$ denote the histories of past actions and stages. Let $x_i (s, \mathcal{H})\in \mathcal{X}_i$ denotes the action of player $i$ at stage $s$, which is dependent on the histories and the current stage. Note that $x_i$ can be random, which corresponds to mixed strategies.


Given $x_1, \cdots, x_N$, the expected discounted payoff for player $i$ starting from stage $s^0$ is: 
\begin{align*}
    \lefteqn{\Pi_i(x_1, \cdots, x_N; s^0)}\\
    &= \E\bigg[\Sigma_{t=0}^\infty \gamma^t M_i(x_1(s^t, \mathcal{H}^t), \cdots, x_N(s^t, \mathcal{H}^t); s^t)\bigg],
\end{align*}
where $\gamma \in (0,1)$ is the discount factor, and the expectation is taken over the randomness in the stage transitions and the possibly mixed strategies.

\subsection{Ex-post MPE}

We consider stochastic games with uncertain payoffs, where for every stage $s \in \mathcal{S}$ and player $i$, the stage game is a polymatrix game. The stage payoff for player $i$ is $M_i(\textbf{x}, s) = \Sigma_{j \in [N], j \neq i} x_i(s) A^{ij}(s) x_j(s)$, where $A^{ij}(s)$ is in an uncertainty set $\mathcal{U}^{ij} (s)$.


An MPE refers to a perfect equilibrium where players' strategies depend only on the current stage but not the histories (i.e., \emph{Markov strategies}), in a stochastic game when the payoffs are fully known. We denote these strategies as $x_i(s)$. We generalize the notion of ex-post equilibrium to MPE.
\begin{definition}
Given a stochastic game with $N$ players, finite stage space $\mathcal{S}$, and uncertain payoffs $A^{ij}(s) \in \mathcal{U}^{ij} (s), \forall s \in \mathcal{S}, \forall i, j \in [N]$, a set of strategies $(x_1^\ast(s), \cdots, x^\ast_N(s))$ is an ex-post MPE if and only if $\forall i, j \in [N], \forall s \in \mathcal{S}, \forall A^{ij}(s) \in \mathcal{U}^{ij} (s)$, $(x_1^\ast(s), \cdots, x^\ast_N(s))$ maximizes the expected discounted payoff for all the players. 
\end{definition}

We provide sufficient conditions for the existence of an ex-post MPE of any $N$-player game with finite stage space and action space by reducing the game to a one-shot finite game. 

For each player $i$ and stage $s$, we create an ``agent'' player $(i, s)$ with action $a(i, s) \in \mathcal{X}_i(s)$. We denote the action profile of all the agent as $\textbf{a} = (a(1, s), \cdots, a(N, s), s \in \mathcal{S})$. The payoff to agent $(i, s)$ is as follows: $Q_{i, s}(\textbf{a}) = \E\bigg[\Sigma_{t=0}^\infty \gamma^t M_i(a(1, s^t), \cdots, a(N, s^t); s^t)|s^0 = s\bigg]$.  


\begin{restatable}[Proof in Appendix \ref{app:proof-polymatrix}]{theorem}{expostMPE}
\label{thm:expost-MPE}
Given that the agent game has an ex-post equilibrium $(a^\ast(i, s), i \in [N], s \in \mathcal{S})$, and defining the strategy for each player $i$ in stage $s \in \mathcal{S}$ as $x_i^\ast(s) = a^\ast(i, s)$, then $(x_i^\ast(s),i \in [N], x \in \mathcal{S})$ is an ex-post MPE for the stochastic game.
\end{restatable}

Note that the agent game is not a polymatrix game due to the stage transitions. Therefore its ex-post equilibrium (when it exists) can not be found using LP(\ref{multiLP}). Solving for ex-post equilibria for multi-player game with more general payoff structures is a topic for future research.

\subsection{Value interval}

When the stage payoffs are fully known, and the games has two players,  \citet{shapley:1953} shows that all two-player, zero-sum stochastic games have a unique value, which can be obtained at an MPE.

Due to the uncertainty in the payoffs, the value of a two-player, zero-sum stochastic game can not be determined fully unless we have access to the complete payoff information. Instead, we show that such game has a value interval where, given any feasible payoffs in the uncertainty set, the value of the full-information game must be in this interval. Further, the interval can be determined by dynamic programming.

\begin{theorem} (Proof in Appendix \ref{app:proof-stochastic}.)
\label{thm:value-interval}
For any two-player, zero-sum stochastic game with finite stage space $\mathcal{S}$ and actions space $\mathcal{X}(s), \mathcal{Y}(s), \forall s \in \mathcal{S}$, let $x(s)^\top A(s)y(s)$ denote the stage payoff for player $x$ in stage $s$, $- x(s)^\top A(s)y(s)$ the stage payoff for player $y$ in stage $s$, and the uncertain payoff $A(s) \in U(s) = \textbf{conv}(\{A_1(s), \cdots, A_{k(s)}(s)\}), \forall s \in \mathcal{S}$. Such a game has a value interval $\mathcal{I}$, where $\forall A(s) \in U(s), \forall s \in \mathcal{S}$, the value of the full-information game with payoff matrices $A(s)$ must lie in $\mathcal{I}$. Further, $\mathcal{I}$ can be found by dynamic programming.
\end{theorem}

\section{Discussion}
\label{sec:discussions}

There are several limitations to our work. First, we do not have necessary conditions for the existence of ex-post equilibria. Our current proof techniques for finding ex-post equilibria rely on bounded polyhedral uncertainty sets and are limited to polymatrix payoffs. In the context of polymatrix games the conditions for ex-post equilibrium are often formulated as optimization problems that use the extremal points of the uncertainty set. For other convex sets, such as norm balls, one could potentially construct a \textit{polyhedral outer approximation}. An important area of future research is to adapt iterative algorithms to find approximate solutions for games with more general payoff structures and uncertainty sets.

Our extensions to stochastic games show that a sufficient condition for existence of ex-post MPE in terms of existence of ex-post equilibrium for the agent game. We leave for future work to determine the necessary conditions for the existence of ex-post MPE. Considering payoff uncertainties at every stage can encode problems of practical relevance to RL such as the use of proxy rewards. Another interesting extension is to consider uncertainty around transition probabilities, which relates to mis-specified MDP problems. 

\bibliography{arxiv_ref}

\clearpage
\appendix


\section{Definitions for Nash Equilibrium, Two-player Zero-sum Games}
\label{app:definitions}

\subsection{Nash equilibrium for N-player polymatrix games}

Under complete information, a set of strategies is a \textit{Nash equilibrium} if each player has no incentive to unilaterally change their strategy given the strategies of all other players. Nash's seminal result shows that for every finite simultaneous-move game, a Nash equilibrium of mixed strategies exists \cite{nash:1950}.

Formally, a Nash equilibrium for the N-player zero-sum polymatrix game is defined as follows: 

\begin{definition}\label{def:Npnash}
(Nash equilibrium for N-player polymatrix game \cite{akarlin}) A mixed strategy profile $(x^{1*},...,x^{N*})$ is a Nash equilibrium for an N-player polymatrix game if for each player $i \in \{1,...,N\}$ and each mixed strategy $x^i \in \Delta_{m_i}$, we have
$$p_i(x^{1*},...,x^{i*},...,x^{N*}) \geq p_i(x^{1*},...,x^i,...,x^{N*})$$
Equivalently, for all $i \in \{1,...,N\}$,
$$x^{i*} \in \arg\max_{x^i \in \Delta_{d_i}} \sum_{[i,j] \in E} (x^i)^T A^{ij} x^{j*}$$
\end{definition}

\citet{cai:2016} show that for zero-sum polymatrix games with finite action space, Nash equilibria can be found efficiently using linear programming.

\subsection{Two-player zero-sum finite games}\label{sec:maxset2pzs}

Two-player zero-sum finite games can be defined as follows: suppose Player 1 has $n$ actions to choose from, and Player 2 has $m$ actions to choose from. The payoffs for each player can be represented by a payoff matrix $A \in \mathbb{R}^{n \times m}$: if Player 1 plays action $i$ and Player 2 plays action $j$, then $A_{ij}$ is the payoff for Player 1 and $-A_{ij}$ is the payoff for Player 2. 

A mixed strategy for Player 1 is defined by a vector in the simplex $x \in \Delta_n$, where $\Delta_n = \{x \in \mathbb{R}^n: x^T \textbf{1} = 1$, $x \geq 0\}$. Each entry $x_i$ is the probability that Player 1 plays action $i$. A similar mixed strategy vector $y \in \Delta_m$ can be defined for Player 2. If Player 1 chooses mixed strategy $x$ and Player 2 chooses mixed strategy $y$, then the expected payoff for Player 1 is $x^TAy$ and the expected payoff for Player 2 is $x^T(-A)y$. Each player simultaneously and independently chooses a strategy to maximize their expected payoff.

A Nash equilibrium for the two-player zero-sum game is defined as follows:

\begin{definition}\label{def:2pnash} (Nash equilibrium for two-player zero-sum game \cite{akarlin})
A pair of strategies $(x^\ast,y^\ast)$ is a Nash equilibrium in a two-player zero-sum game with payoff matrix $A$ if 
$$ x^\ast \in \arg\max_{x \in \Delta_n} x^T A y^\ast, \;\; y^\ast \in \arg\max_{y \in \Delta_m} (x^\ast)^T (-A) y$$
or equivalently,
$$\min_{y \in \Delta_m} (x^\ast)^T A y = (x^\ast)^TAy^\ast = \max_{x \in \Delta_n} x^T A y^\ast$$
\end{definition}
Definition \ref{def:2pnash} shows that the Nash equilibrium strategies $x^\ast$ and $y^\ast$ are best response strategies to each other: given that Player 2 plays with strategy $y^\ast$, the strategy $x^\ast$ maximizes Player 1's expected payoff. Likewise, given that Player 1 plays with strategy $x^\ast$, the strategy $y^\ast$ maximizes Player 2's expected payoff. 

\section{Graphical representation of the consolidated payoff matrix $R$}\label{app:r_matrix}
\begin{center}
\resizebox{.6\textwidth}{.6\textwidth}{%
\begin{tikzpicture}
\draw[step=2cm,gray, thin, dotted] (0,0) grid (12,12);
\draw[step=2cm,gray, thin] (-0,-0) grid (2,2);
\draw[step=2cm,gray, thin] (0,6) grid (6,12);
\draw[step=2cm,gray, thin] (8,0) grid (12,4);
\draw[step=2cm,gray, thin] (10,10) grid (12,12);
\draw (0, 0) rectangle (12,12);
\node at (1, 9) {\Large$A^{21}$};
\node at (3, 11) {\Large$A^{12}$};
\node at (1, 7) {\Large$A^{31}$};
\node at (5, 11) {\Large$A^{13}$};
\node at (3, 7) {\Large$A^{32}$};
\node at (5, 9) {\Large$A^{23}$};
\node at (1, 1) {\Large$A^{N1}$};
\node at (11, 11) {\Large$A^{1N}$};
\node at (9, 1) {\Large$A^{N,N-1}$};
\node at (11, 3) {\Large$A^{N-1,N}$};
\node at (7, 5) {\LARGE{$\ddots$}};
\foreach \x in {0,2,4 , 8, 10} {
\node at (\x+1, 11-\x) {\Large{$0$}};
\draw (\x,10-\x) rectangle (\x+2,12-\x); }

\node at (1, 12.3) {$d_1$ cols};
\node at (3, 12.3) {$d_2$ cols};
\node at (5, 12.3) {$d_3$ cols};
\node at (8, 12.3) {\LARGE{\ldots}};
\node at (11, 12.3) {$d_N$ cols};

\node at (-0.3, 11) {\rotatebox{90}{$d_1$ rows}};
\node at (-0.3, 9) {\rotatebox{90}{$d_2$ rows}};
\node at (-0.3, 7) {\rotatebox{90}{$d_3$ rows}};
\node at (-0.3, 4) {\rotatebox{90}{\LARGE{\ldots}}};
\node at (-0.3, 1) {\rotatebox{90}{$d_N$ rows}};
\end{tikzpicture} } 
\end{center}

\section{Proofs for Section \ref{sec:polymatrix}}
\label{app:proof-polymatrix}

\LPexpostn*
 Any ex-post equilibrium of the zero-sum polymatrix game with uncertainty sets $\mathcal{U} = \textbf{conv}{\{R_1, R_2,...R_K\}}$ is an optimal solution of the LP problem (\ref{multiLP}). Conversely, any optimal solution for the LP problem (\ref{multiLP}) is an ex-post equilibrium of the game.

We prove the above theorem through the following two lemmas, where in lemma \ref{lem:LP-expost-n} we consider $\mathcal{U} = {\{R_1, R_2,...R_K\}}$, and in lemma \ref{lem:LP-expost-n-conv} we generalize the result to $\mathcal{U} = \textbf{conv}{\{R_1, R_2,...R_K\}}$

\begin{lemma} \label{lem:LP-expost-n}
 Any ex-post equilibrium of the zero-sum polymatrix game with uncertainty sets $\mathcal{U} = {\{R_1, R_2,...R_K\}}$ is an optimal solution of the LP problem (\ref{multiLP}). Conversely, any optimal solution for the LP problem (\ref{multiLP}) is an ex-post equilibrium of the game. 
\end{lemma}

\begin{proof}
Consider any feasible solution $x,w$: 
\begin{align}
\label{expostLP-nplayer}
\begin{split}
    \sum_{l=1}^K \sum_{i=1}^N w_i^l \geq &\sum_{l=1}^K \sum_{i=1}^N \max_{a_i} e_{a_i}^T R_lx \\
    =&\sum_{l=1}^K\sum_{i=1}^N \max_{y^i \in \tilde{\Delta}_{d_i}} (y^i)^T R_lx \\
    \geq& \sum_{l=1}^K \sum_{i=1}^N (x^i)^T R_lx\\
    =&\ 0
\end{split}
\end{align}
where $y^i\in \mathbb{R}^{\sum_{i \in N}d_i}$ is a vector where the sub-vector $y_{[i]} \in \Delta_{d_i}$, and all other entries are zero; $x^i$ is the vector generated from $x$, where the sub-vector $x_{[i]}$ is kept the same as $x$ and all other entries are zero. The first step comes from the fact that $x,w$ is feasible; the second step comes from the fact that the objective is linear in $e_{a_i}$; the last step follows since for every $R_i, i \in [K]$, the game is zero-sum. 

Consider an ex-post equilibrium $x^\ast$: the second inequality in (\ref{expostLP-nplayer}) becomes equality by the definition of the ex-post equilibrium. By setting $(w_i^l)^\ast = ((x^\ast)^i)^T R_l x^\ast, \forall i \in [N], l \in [K]$, the first inequality in (\ref{expostLP-nplayer}) becomes equality. Therefore, $\sum_{l=1}^K \sum_{i=1}^N (w_i^l)^\ast = 0$, which means $(x^\ast, w^\ast)$ is an optimal solution to the LP.

Consider an optimal solution $(x^\ast, w^\ast)$ to the LP problem (\ref{multiLP}), if the game has an ex-post equilibrium, then the objective function of the optimal solution is zero. Therefore, the second inequality in (\ref{expostLP-nplayer}) becomes equality. Therefore,  $x^\ast$ is an ex-post equilibrium.
\end{proof}

Next, we generalize the result to the larger uncertainty set $\mathcal{U} = \textbf{conv}{\{R_1, R_2,...R_K\}}$. 

\begin{lemma}
 \label{lem:LP-expost-n-conv}
Any ex-post solution to the N-player zero-sum polymatrix game with uncertainty set $\mathcal{U} = {\{R_1, R_2,...R_K\}}$ is an ex-post equilibrium for the N-player zero-sum polymatrix game with uncertainty set $\mathcal{U} = \textbf{conv}{\{R_1, R_2,...R_K\}}$.
\end{lemma}

\begin{proof}
Denote the game with $\mathcal{U} = {\{R_1, R_2,...R_K\}}$ as Game (I), and the game with $\mathcal{U} = \textbf{conv}{\{R_1, R_2,...R_K\}}$ as Game (II). 
Consider an ex-post equilibrium $x$ for Game (I). Then by the definition of the ex-post equilibrium, $\forall i \in [N],\ \forall l \in [K]$
\begin{align*}
    x^i \in \arg\max_{y^i \in \tilde{\Delta}_{d_i}} (y^i)^T R_l x
\end{align*}

Consider $R = \sum_{l = 1}^K z_l R_l$ where $z \in \Delta_k$.  We have that
$$\arg\max_{y^i \in \tilde{\Delta}_{d_i}} (y^i)^T \left(\sum_{l = 1}^K z_l R_l\right) x = \arg\max_{y^i \in \tilde{\Delta}_{d_i}} \sum_{l = 1}^K z_l \left((y^i)^T R_l x\right)$$
it follows that: $$x \in \arg\max_{y^i \in \tilde{\Delta}_{d_i}} \sum_{l = 1}^K z_l \left((y^i)^T R_l x \right) $$
Therefore: $$x \in \arg\max_{y^i \in \tilde{\Delta}_{d_i}} (y^i)^T \left(\sum_{l = 1}^K z_l R_l\right) x$$
Hence, $x$ is an ex-post equilibrium.
\end{proof}

\section{Proofs for Section \ref{sec:stochastic}}
\label{app:proof-stochastic}

In this appendix, we provide proof details for results on stochastic games.

\expostMPE*
Given the agent game has an ex-post equilibrium $(a^\ast(i, s), i \in [N], s \in \mathcal{S})$, define the strategy for each player $i$ in stage $s \in \mathcal{S}$ as $x_i^\ast(s) = a^\ast(i, s)$, then $(x_i^\ast(s),i \in [N], x \in \mathcal{S})$ is an ex-post MPE for the stochastic game.

\begin{proof}
Note that the agent game has finitely many players and actions. For each agent $(i, s)$, the payoff $Q_{i, s}(\textbf{x})$ lies in an uncertainty set defined as $$G_{i, s} = \bigg\{Q_{i, s}(\textbf{x}): Q_{i, s}(\textbf{x}) = \E\bigg[\Sigma_{t=0}^\infty\Sigma_{j \in [N], j \neq i}  \gamma^t  x_i(s) A^{ij}(s) x_j(s)|s^0 = s\bigg], A^{ij}(s) \in \mathcal{U}^{ij} (s) \bigg\}$$
Given this finite agent game has an ex-post equilibrium $\textbf{a}^\ast = (a^\ast(i, s), i \in [N], a \in \mathcal{S})$, define the strategy for each player $i$ in stage $s \in \mathcal{S}$ as $x_i^\ast(s) = a^\ast(i, s)$. By the definition of ex-post equilibrium, 
\begin{align*}
    a^\ast(i, s) &= \argmax_{a(i, s)} \;  Q_{i, s}(a(i, s), \textbf{a}^\ast \setminus (i, s)), \quad \forall Q_{i, s} \in G_{i, s} \\
    \iff x_i^\ast(s) &= \argmax_{x_i(s)}  \; \Pi_i(x_i(s), \textbf{x}^\ast_{-i}(s); s), \quad \forall A^{ij}(s) \in \mathcal{U}^{ij}(s), i, j \in [N]
\end{align*}
where we denote $\textbf{x}^\ast_{-i}(s)$ as the set of strategies $x_i^\ast(s)$ for the players except player $i$. Obviously $x_i^\ast(s)$ only depends on the current stage. Given $\textbf{x}^\ast_{-i}(s)$, $x_i^\ast(s)$ maximizes the expected discounted payoff of player $i$ among all Markov strategies. Since for $\forall A^{ij}(s) \in \mathcal{U}^{ij}(s), i, j \in [N]$, if all other players except player $i$ are playing a Markov strategy, player $i$ has a best response in stage $s$ which is a Markov strategy (by identifying $\argmax_{x_i \in \mathcal{X}_i(s)}\Pi_i(x_i, \textbf{x}_{-i}(s); s)$), we conclude that $(x_i^\ast(s),i \in [N], x \in \mathcal{S})$ is an ex-post MPE for the stochastic game.
\end{proof}

\begin{theorem} \label{thm:value-interval}
For any two-player, zero-sum stochastic game with finite stage space $\mathcal{S}$ and actions space $\mathcal{X}(s), \mathcal{Y}(s), \forall s \in \mathcal{S}$, denote $x(s)^\top A(s)y(s)$ the stage payoff for player $x$ in stage $s$, $- x(s)^\top A(s)y(s)$ the stage payoff for player $y$ in stage $s$, and the uncertain payoff $A(s) \in U(s) = \textbf{conv}(\{A_1(s), \cdots, A_{k(s)}(s)\}), \forall s \in \mathcal{S}$. Such a game has a value interval $\mathcal{I}$, where $\forall A(s) \in U(s), \forall s \in \mathcal{S}$, the value of the full-information game with payoff matrices $A(s), \forall s \in \mathcal{S}$ must be in $\mathcal{I}$. Further, $\mathcal{I}$ can be computed by dynamic programming.
\end{theorem}
\begin{proof}
We prove by value iteration. Denote $\gamma \in (0,1)$ as the discount factor and $\Prob(s'|s, x, y)$ as the transition probability. First, we pick arbitrary functions $\alpha: \mathcal{S} ->\R$, $\beta: \mathcal{S} ->\R$ such that $\alpha(s) > \beta(s), \forall s \in \mathcal{S}$. For each $s \in \mathcal{S}$, define matrix $ M_\alpha(s)$,  $ M_\beta(s)$  as: $\forall x \in \mathcal{X}(s), y \in \mathcal{Y}(s)$, 
\begin{align*}
    M_\alpha(s)(x, y) &= x^\top A(s)y + \gamma \Sigma_{s'\in \mathcal{S}}\Prob(s'|s, x, y)\alpha(s')\\
    M_\beta(s)(x, y) &= x^\top A(s)y + \gamma \Sigma_{s\in \mathcal{S}}\Prob(s'|s, x, y)\beta(s')
\end{align*}
Since $A(s) \in U(s) = \textbf{conv}(\{A_1(s), \cdots, A_{k(s)}(s)\})$, $x^\top A(s)y \in [\lambda(x, y, s), \delta(x, y, s)]$, where we use $\lambda(x, y, s)$ and $\delta(x, y, s)$ to denote the lower and upper bound of $x^\top A(s)y$, given $x, y$ and $A(s) \in U(s)$. 

For each stage $s$, consider the two-player zero-sum finite game with payoff matrix 
$$ \bar{ M}_\alpha(s)(x, y) = \delta(x, y, s) + \gamma \Sigma_{s'\in \mathcal{S}}\Prob(s'|s, x, y)\alpha(s')$$ $x \in \mathcal{X}(s), y \in \mathcal {Y}(s)$. Since $\mathcal{X}(s), \mathcal{Y}(s)$ are finite, this game has a value. Denote operator $T$ as $(T\alpha)(s) = \text{val}(\bar{ M}_\alpha(s))$. Given $\alpha$ and $\alpha'$, 
\begin{align*}
\lefteqn{||T\alpha - T\alpha'||_{\infty}} \\
&= \max_{s \in \mathcal{S}}|T\alpha(s) - T\alpha'(s)|\\
&= \max_{s \in \mathcal{S}}|\text{val}(\bar{ M}_\alpha(s)) - \text{val}(\bar{ M}_{\alpha'}(s))|\\
&\leq \max_{s \in \mathcal{S}}\max_{x, y} \gamma |\Sigma_{s'\in \mathcal{S}}\Prob(s'|s, x, y)\alpha(s') -\Sigma_{s'\in \mathcal{S}}\Prob(s'|s, x, y)\alpha'(s')| \\
&= \gamma \max_{s \in \mathcal{S}}|\alpha(s) - \alpha'(s)| \\
&= \gamma||\alpha - \alpha'||_{\infty}
\end{align*}

Since $\gamma \in (0,1)$, $T$ is a contraction. Thus, $\alpha_k = (T\alpha_{k-1})$ converges to a unique limit $\alpha^\ast: \alpha^\ast = T\alpha^\ast$ as $k -> \infty$. Similarly, define the two-player zero-sum finite game with payoff matrix 
$$ \underline{ M}_\beta(s)(x, y) = \lambda(x, y, s) + \gamma \Sigma_{s'\in \mathcal{S}}\Prob(s'|s, x, y)\beta(s')$$ $x \in \mathcal{X}(s), y \in \mathcal {Y}(s)$, and operator $(J\beta)(s) = \text{val}(\underline{J}_\beta(s))$. By a similar argument, $\beta_k = (J\beta_{k-1})$ converges to a unique limit $\beta^\ast: \beta^\ast = J\beta^\ast$ as $k -> \infty$. Notice that by construction, $ \bar M_\alpha(s)(x, y) > \underline M_\beta(s)(x, y), \forall s \in \mathcal{S}, \forall x \in \mathcal{X}(s), y \in \mathcal {Y}(s)$, which is preserved by the value iteration: $\bar M_{(T\alpha(s))}(x, y) > \underline M_{(J\beta(s))}(x, y), \forall s \in \mathcal{S}, \forall x \in \mathcal{X}(s), y \in \mathcal {Y}(s)$. Therefore, $\alpha^\ast(s) \geq \beta^\ast(s), \forall s \in \mathcal{S}$. 

For any feasible true payoffs $\{A(s), A(s) \in U(s), \forall s \in \mathcal{S}\}$, by construction we have $  \underline M_\beta(s)(x, y) \leq M_\alpha(s)(x, y) \leq \bar M_\alpha(s)(x, y), \forall s \in \mathcal{S}$. Doing value iteration on $\alpha$ with the true payoffs yields the value of the game $\alpha^\ast_{true}$ with full payoff information $\{A(s), A(s) \in U(s), \forall s \in \mathcal{S}\}$. Since the monotonicity is preserved by value iteration, we have $\alpha^\ast_{true}(s)  \in \mathcal{I}(s), \forall s \in \mathcal{S}$, where $\mathcal{I}(s) = [\beta^\ast(s), \alpha^\ast(s)]$. 
\end{proof}

\section{Further Discussions on Maximal Uncertainty Set for Ex-post Equilibrium}
\label{app:maximal-uncertainty}

Beyond finding ex-post equilibria for a given uncertainty set, we provide characterizations of the \textit{maximal uncertainty sets} for a given Nash equilibrium strategy profile. We define the maximal uncertainty set to be the largest set of payoffs for which the given strategy profile stays at a Nash equilibrium. For better intuition, we start with two-player zero-sum games, and then extend the analysis to zero-sum polymatrix games.

Informally, a tuple of strategies is \textit{stable} at an ex-post equilibrium if it is indifferent to the choice of payoff. It means that when changing the payoff, the players have no incentive to unilaterally change their strategies. A first step towards understanding when ex-post equilibrium is attainable is describing the types of uncertainty sets that admit one. We will frame this discussion in terms of \textit{maximal uncertainty sets} corresponding to a tuple of strategies. 

Formally, given a tuple of Nash equilibrium strategies to a nominal game, we will describe a maximal uncertainty set $\mathcal{U}$ such that the given strategies are at ex-post equilibrium with respect to $\mathcal{U}$.

\subsection{Maximal uncertainty in two-player zero-sum games}
First we consider the case of a  two-player zero-sum finite game. Similar to the N-player polymatrix games, we provide the definition of an ex-post equilibrium for a two-player zero-sum finite game.

An ex-post equilibrium for a two-player zero-sum finite game as follows:
\begin{definition}\label{def:2p_expost}
(ex-post equilibrium for two-player zero-sum game) 
A tuple of strategies $(x^\ast, y^\ast)$ is an ex-post equilibrium for the two-player zero-sum game with uncertainty set $\mathcal{U}$ if $\forall A \in \mathcal{U}$, the following holds:
\begin{align*}
x^\ast &\in \arg\max_{x\in \Delta_n} x^T A y^\ast\\
 y^\ast &\in \arg\max_{y\in \Delta_m} x^{\ast T} (-A) y
\end{align*}
\end{definition}

Let $x^\ast \in \Delta_n$ and $y^\ast \in \Delta_m$ be the mixed strategy Nash equilibrium of a nominal game with payoff $A \in \mathbb{R}^{n\times m}$. The maximal uncertainty set is defined as follows:
\begin{definition}
Given a nominal zero-sum  two-player game with payoff matrix $\tilde{A}$ and Nash equilibrium $(x^\ast, y^\ast)$, a set $\mathcal{A}(x^\ast, y^\ast)$ is the maximal uncertainty set if $\forall A\in \mathcal{A}(x^\ast, y^\ast)$:
\begin{align*}
    x^\ast & \in \arg\max_{x \in \Delta_n} x^T A y^\ast; \; \text{ and } \; y^\ast \in \arg\max_{y \in \Delta_m} (x^\ast)^T (- A) y
\end{align*}
and for all $A \notin \mathcal{A}(x^\ast, y^\ast)$:
\begin{align*}
    x^\ast & \notin \arg\max_{x \in \Delta_n} x^T A y^\ast ;\ \text{ or } \; y^\ast \notin \arg\max_{y \in \Delta_m} (x^\ast)^T (- A) y
\end{align*}
\end{definition}
As a consequence, an uncertain two-player zero-sum finite game with payoff uncertainty $\mathcal{U}$ is at ex-post equilibrium for $(x^\ast, y^\ast)$ if and only $\mathcal{U} \subseteq \mathcal{A}(x^\ast, y^\ast)$.\\
\begin{lemma}\label{lem:maxset2pzs}
 A pair $(x^\ast, y^\ast)$ of best-response strategies to game with payoff $\tilde{A}$ is best response to a game with payoff ${A}$ if ($A, x^\ast, y^\ast)$ satisfy the following conditions:
 \begin{align}
     e_i^T A y^\ast &= c \quad \text{if } x_i^\ast > 0, \forall i = 1, 2, \ldots, n  \label{maxset2pzs1} \\
     e_i^T A y^\ast &\leq c \quad \text{if } x_i^\ast = 0, \forall i = 1, 2, \ldots, n \label{maxset2pzs2} \\
     x^{\ast^T} A e_j &= c \quad \text{if } y_j^\ast > 0, \forall j = 1, 2, \ldots, m \label{maxset2pzs3}\\
     x^{\ast^T} A e_j &\geq c \quad \text{if } y_j^\ast = 0, \forall j = 1, 2, \ldots, m \label{maxset2pzs4}
 \end{align}
 Where $e_i, e_j$ are the appropriate unit vector such that $e_i^T A $ is the $i-$th row and $A e_j$ is the $j$th column of $A$.
\end{lemma} 

\begin{proof}
To show  \ref{maxset2pzs1}, assume by contradiction: 
$$\exists i, i' \quad \text{s.t.} \quad  x_i^\ast > 0,\; \; x_{i'}^\ast >0 \quad  \text{and} \quad e_i^T A y^\ast > e_{i'}^T{A}y^\ast$$ 
Player 1 is incentivized to unilaterally change the strategy to $x^{\ast\ast}$ by transferring the probability associated with strategy $i'$ to strategy $i$, i.e. $x^{\ast\ast}_{i}=x^\ast_{i}+x^\ast_{i'}$, and $x^{\ast\ast}_{i'} = 0$.

Similarly for \ref{maxset2pzs2}, if $x^\ast_i = 0$, then $\exists i'$ such that $x_{i'} > 0$. Assume by contradiction that:
$$\exists i, i' \quad \text{s.t.} \quad  x_i^\ast = 0,\; \; x_{i'}^\ast >0 \quad  \text{and} \quad e_i^T{A}y^\ast > e_{i'}^T{A}y^\ast = c$$
Player 1 is incentivized to change the strategy to $x^{\ast\ast}$ by swapping the probability associated with strategy $i'$ to strategy $i$, i.e. $x^{\ast\ast}_{i}=x^\ast_{i'}$, and $x^{\ast\ast}_{i'} = 0$. By a symmetric argument we can show \ref{maxset2pzs3} and \ref{maxset2pzs4}. The constant $c$ is the same across all identities.
\end{proof} 
By Lemma \ref{lem:maxset2pzs} we can conclude that the \textit{maximal uncertainty set} associated to a nominal zero-sum  two-player game: $(x^\ast, y^\ast, \tilde{A})$ is the set of matrices that satisfy \ref{maxset2pzs1}-\ref{maxset2pzs4}. 

\subsection{N-player zero-sum polymatrix games}
We can extend the argument in \ref{lem:maxset2pzs} to zero-sum polymatrix games. Let $A^{ij}\in \mathbb{R}^{d_i\times d_j}$ be the payoff matrix for player $i$ when playing against player $j$. Suppose we have $N$ players with respective mixed strategies $\{x^1,...,x^N\}$ where $ x^i \in \Delta_{d_i},\ \forall i \in [N]$. The strategy $x^i$ chosen by player $i$ is used across all the pairwise games that $i$ participates in. Let $E$ be the adjacency set. The payoff of player $i$ is the sum of payoffs in all the games as defined as in \ref{def: nzerosum}:
\begin{align*}
    p_i(x) = \sum_{[i, j]\in E}x^i A^{ij}x^j
\end{align*}

When there is no uncertainty in the payoff matrices, a tuple of strategies $\{x^{1\ast},...,x^{N\ast}\}$ is at Nash equilibrium if the following holds:
\begin{align*}
    x^{i\ast} \in \arg \max_{x^i \in \Delta_{d_i}} p(x^{1\ast}, \ldots, x^i, \ldots x^{N\ast}) \quad \forall i \in 1, 2, \ldots N
\end{align*}

We define the maximal uncertainty set in a polymatrix game as follows:
\begin{definition}
Given a nominal polymatrix games with $N$ players, adjacency set $E$ and Nash equilibrium strategy tuple $X^\ast = (x^{1\ast}, x^{2\ast}, \ldots , x^{N\ast} )$, the maximal uncertainty set for the game is the product set: $\mathcal{A}(X^\ast) = \mathcal{A}^{1, 2}(X^\ast)\times \mathcal{A}^{2, 1}(X^\ast) \times \ldots \times \mathcal{A}^{N, N-1}(X^\ast)$ is the maximal set that satisfies the following:
\begin{align*}
 x^{i\ast} \in \arg \max_{x^i\in \Delta_{d_i}}\sum_{[i, j]\in E}x^i {A}^{ij}x^{j\ast}; \quad \forall i, \forall {A}^{ij}\in \mathcal{A}^{ij}
\end{align*}
\end{definition}
As a consequence, an uncertain polymatrix game with payoff uncertainty $\mathcal{U}=\mathcal{U}^{1,2}\times \ldots \times \mathcal{U}^{N, N-1}$ admits ex-post equilibrium iff $\mathcal{U}\subseteq \mathcal{A}(X^\ast)$ (i.e. $\forall i\ne j, \mathcal{U}^{i, j} \subseteq \mathcal{A}^{i, j}(X^\ast)$).

\begin{lemma}\label{lem:maxset}
 A tuple $(x^{1\ast}, \ldots,  x^{N\ast})$ is best response to a polymatrix game with payoffs $\tilde{A}^{ij}$ if the following conditions are satisfied for each player:
 \begin{align}
     e_k^T\sum_{[i, j]\in E}{A}^{ij} x^{j\ast} &= c_i, \quad \text{if } x^{i\ast}_k > 0, \forall k = 1, 2, \ldots, d_i  \label{maxset1} \\
     e_k^T\sum_{[i, j]\in E}{A}^{ij} x^{j\ast} &\leq c_i, \quad \text{if } x^{i\ast}_k = 0, \forall k = 1, 2, \ldots, d_i \label{maxset2}
 \end{align}
 Where $e_k$ is the $k$-th unit vector, hence $ e_k^T {A}^{ij}$ denotes is the $k-$th row in the payoff matrix associated to players $(i,j)$. Lastly $x^{i\ast}_k$ is the $k$-th entry in the strategy vector of player $i$.
\end{lemma} 

\begin{proof}
The proof is similar to the proof for \ref{lem:maxset2pzs}. To show \ref{maxset1} assume by contradiction that $\exists i, k, k'$ such that:
$$x_{k}^{i\ast}>0, \; x_{{k'}}^{i\ast}>0\quad \text{and} \quad e_k^T \sum_{[i, j]\in E}{A}^{ij}x^{j\ast} > e_{k'}^T \sum_{[i, j]\in E}{A}^{ij}x^{j\ast} $$
player $i$ is incentivized to change the strategy to $x^{i\ast\ast}$ by transferring the probability associated to strategy $k'$ to strategy $k$, (i.e. $x^{i\ast\ast}_{{k}} = x^{i\ast}_{{k}} + x^{i\ast}_{{k'}}$ and $x^{i\ast\ast}_{{k'}} = 0$ ).
To show \ref{maxset2}, assume by contradiction that $\exists i,  k, k'$ such that:
$$x_k^{i\ast}=0,\; x_{k'}^{i\ast}>0\quad \text{and} \quad e_k^T \sum_{[i, j]\in E}{A}^{ij}x^{j\ast} > e_{k'}^T \sum_{[i, j]\in E}{A}^{ij}x^{j\ast} $$
player $i$ is incentivized to change the strategy to $x^{i\ast\ast}$ by swapping the probability associated to strategy $k'$ to strategy $k$, (i.e. $x^{i\ast\ast}_{{k}} = x^{i\ast}_{{k'}}$ and $x^{i\ast\ast}_{{k'}} = 0$ ).
\end{proof}

Note that it is not necessary for an arbitrary game from the maximimal uncertainty game to be zero-sum as in Definition \ref{def: nzerosum}. This means that an optimal strategy of a nominal zero-sum polymatrix game remains optimal in an ex-post sense for the maximal uncertainty set defined above even when the zero-sum condition is violated. Lemma \ref{lem:maxset2pzs} and \ref{lem:maxset}  define a set of restrictions on payoffs that preserve the best-response property of a tuple of strategies.

\end{document}